\documentclass[authoryear,preprint,3p,review,11pt]{elsarticle}
\journal{Elsevier}

\usepackage{amssymb}
\usepackage{amsthm}
\usepackage{amsmath}
\usepackage{amssymb}
\usepackage{amsfonts}
\usepackage{multirow}
\usepackage{amsthm}

\usepackage{graphicx}
\usepackage{slashbox}
\usepackage{natbib}

\usepackage{todonotes}
\usepackage{comment}
\newtheorem{theorem}{Theorem}
\newtheorem{lemma}[theorem]{Lemma}

\theoremstyle{definition}

\newcommand{\F}{\mathbb{F}}

\newcommand{\E}{\mathbb{E}}

\newcommand{\KS}{{\rm KS}}
\newcommand{\KT}{{\rm KT}}

\begin{document}
\makeatletter
\def\ps@pprintTitle{%
  \let\@oddhead\@empty
  \let\@evenhead\@empty
  \let\@oddfoot\@empty
  \let\@evenfoot\@oddfoot
}
\makeatother

\begin{frontmatter}



\author[label1]{Tommaso Lando\corref{c1}} 
\fntext[label1]{Department of Economics, University of Bergamo, Italy; Department of Finance, V\v{S}B-TU Ostrava, Czech Republic}
\ead{tommaso.lando@unibg.it}

\author[label2]{Idir Arab} 
\fntext[label2]{CMUC, Department of Mathematics, University of Coimbra, Portugal}

\author[label2]{Paulo Eduardo Oliveira} 

\cortext[c1]{Corresponding author}

\title{Nonparametric inference about increasing odds rate distributions}




\begin{abstract}
To improve nonparametric estimates of lifetime distributions, we propose using the increasing odds rate (IOR) model as an alternative to other popular, but more restrictive, ``adverse ageing'' models, such as the increasing hazard rate one. This extends the scope of applicability of some methods for statistical inference under order restrictions, since the IOR model is compatible with heavy-tailed and bathtub distributions. We study a strongly uniformly consistent estimator of the cumulative distribution function of interest under the IOR constraint. Numerical evidence shows that this estimator often outperforms the classic empirical distribution function when the underlying model does belong to the IOR family. We also study two different tests, aimed at detecting deviations from the IOR property, and we establish their consistency. The performance of these tests is also evaluated through simulations.
\end{abstract}

\begin{keyword}
hazard rate\sep  heavy tails \sep nonparametric test \sep odds



\end{keyword}

\end{frontmatter}
\section{Introduction}
In statistics, the problem of estimating a cumulative distribution function (CDF) $F$, often {depends on any prior information on $F$}. {This is, for example, the idea that} gives rise to parametric inference. On the other hand, no information on $F$ leads to the most basic nonparametric estimator of $F$, namely, the empirical CDF $\F_n$. The {approach based on} shape restricted inference represents an intermediate case, in which it is assumed that $F$ belongs to some important{, and as broad as possible,} nonparametric family of distributions, satisfying some shape constraints. It is obvious that, using a stronger shape constraint, one may expect a larger estimation improvement upon $\F_n$, but, on the other hand, less applicability and higher risk, if the constraint is false. Thus, these shape restrictions should be chosen to reasonably conform to the data, and validated through statistical testing.

In reliability and survival analysis, one is generally interested in random lifetimes, so it is typically assumed that distributions may exhibit ``adverse ageing'', vaguely meaning that ageing has a negative effect on lifetime \citep{marshall2007}. The classic (but not unique) way of translating this intuitive notion into a mathematical language is assuming that $F$ has an \textit{increasing hazard rate} (IHR). For this reason, nonparametric inference for IHR distributions represents an extremely relevant case, which has been studied by \cite{grenander}, \cite{marshallmle} or \cite{rao}, {among} many other authors. {However}, the IHR assumption is sometimes considered to be too strong. For example, it is not compatible with heavy-tailed ({as it requires} the existence of all moments) and \textit{bathtub} distributions, namely, distributions that exhibit a U-shaped hazard rate \citep{marshall2007}. The solution to this problem may be using a weaker shape restriction, which is still coherent with the ``adverse ageing'' notion. For example, \cite{Wang1987}, \cite{rojo} or \cite{elbarmi} estimate distributions which exhibit an IHR, {but only} on average (IHRA). Although the IHRA class contains the IHR class, this family still requires the existence of all moments and does not contain any bathtub model. Therefore, the problem of estimating $F$ within a weaker ``adverse ageing'' setting, compatible with heavy-tailed and bathtub distributions, is particularly interesting.

In a recent paper, \cite{oddspaper} proposed an ageing class that is based on the {monotonicity of the} \textit{odds rate} (OR), instead of the hazard rate. In fact, ageing classes are typically defined in terms of some suitable stochastic ordering constraint, taking the exponential distribution as the classic benchmark for ``no ageing''. This is due to the well-known ``lack of memory'' property. However, a wider family than the IHR may be constructed by replacing the exponential with a different reference distribution, providing an alternative interpretation of ``no ageing''. In particular, as discussed in \cite{oddspaper}, the log-logistic distribution (with shape parameter equal to 1), hereafter referred to as LL(1), may be a suitable benchmark, as it satisfies a ``multiplicative lack of memory'' property \citep{galambos}, and it is the only distribution such that the \textit{odds of failure} by time $x$, that is the probability of failure over survival at $x$, has a constant growth rate with respect to time. Moreover, unlike the exponential distribution, the LL(1) {distribution} has an infinite expectation, which is also more coherent with {a} literal {intuitive} translation of the ``no ageing'' concept. The family of distributions that are dominated by the LL(1) with respect to the \textit{convex transform order} of \cite{vanzwet1964} is called the \textit{increasing odds rate} (IOR) family. Besides having some interesting properties, the IOR class has a wide range of applicability{:} it does not require the existence of moments, it contains all IHR distributions, plus some important heavy-tailed and bathtub ones (see Section 2.1 of \cite{oddspaper}). Accordingly, the objective of this paper is to obtain an estimator of $F$ {that satisfies the IOR assumption}, {hence improving} upon $\F_n$ under weak conditions. In Section 3, we construct such an estimator, denoted as $\widetilde{\F}_n$, using similar methods as those discussed in the books of \cite{barlow1971} or \cite{robertson1988}. $\widetilde{\F}_n$ is strongly uniformly consistent, and simulations show that it is generally more accurate than $\F_n$, provided that the IOR assumption holds. On the other hand, it is clear that the method does not work properly when the IOR assumption is false, therefore it is important to decide, based on statistical testing, whether the data support the IOR assumption or not. In Section 4, we obtain two different nonparametric tests of the IOR null hypothesis. We study the properties of these tests, establish their consistency, and compare their performances simulating the values of the power function under some critical alternatives.

\section{The IOR property}
Let us begin with some preliminary notations. Throughout this paper, ``increasing'' signifies ``non-decreasing'' and ``decreasing'' signifies ``non-increasing''. We define the \textit{greatest convex minorant} (GCM) of a function $g$, denoted as $g_{cx}$, as the largest convex function that does not exceed $g$. Similarly, the \textit{least concave majorant} (LCM) of $g$, denoted as $g_{cv}$, is the smallest concave function which is larger than or equal to $g$. {Throughout this paper,} $X$ {denotes} a random variable with an absolutely continuous CDF $F$ and probability density function (PDF) $f$. {Moreover}, for the sake of simplicity, any map $\phi$ which depends on a CDF $F$, will be denoted as $\phi_F$ or just $\phi$, whenever it is clear from the context.

We will now recall the definition and properties of the IOR family. The reader is referred to \cite{oddspaper} for a more detailed discussion.

Let us denote with $L(x)=\frac x{1+x}$, $x>0$, the CDF of the LL(1) distribution. The function $$\Lambda_F(x)=L^{-1}\circ F(x)=\frac{F(x)}{1-F(x)}$$ defines the odds of failure by time $x$. The function $\Lambda_F$ is obviously increasing, so it is natural to get interested in its growth rate. In fact, one may expect that $F$ exhibits ``no ageing'', in some sense, when the \textit{odds rate} of $F$, defined as
$${\lambda_F(x)=\Lambda_F'(x)=\frac{f(x)}{(1-F(x))^2}},$$
is constant. Similarly, an increase of $\lambda$ means that the probability of failure over survival is accelerating with respect to time, suggesting an ``adverse ageing'' scenario. Therefore, we are interested in distributions with an increasing behaviour of $\lambda$, or, equivalently, with a convex odds function $\Lambda$.
This family, formally $\mathcal{F}_{IOR}=\{F:\Lambda_F\text{ is convex}\}$, is referred to as the IOR class.

{Note that the convexity of $\Lambda$ does not imply the existence of a density. Indeed, this convexity allows for a single point mass of $F$ at the right endpoint of the support. However, for simplicity of the presentation, we shall assume the existence of a PDF}.

It is easy to see that the IOR class contains the well known IHR class, {which may be characterized as} $\{F:-\ln(1-F)\text{ is convex}\}$. These two classes may be equivalently defined in terms of the convex transform order, requiring convexity of $G^{-1}\circ F$, where $G$ is some {suitable} (and fixed) reference CDF.
In {general}, $G$ {dominates} $F$ in the convex transform order, denoted as {$G\geq_c F$}, if $G^{-1}\circ F$ is convex \citep{shaked2007}. The IHR class is obtained by choosing $G(x)=\mathcal{E}(x)=1-e^{-x}$, $x>0$, namely, the unit exponential distribution. The main reason for choosing $\mathcal{E}$ as a benchmark is that it satisfies the ``lack of memory'' property. {Likewise}, for $G={L}$, we obtain the IOR class, and it is also easy to see that, if $X\sim L$, then the shifted random variable $X+1$ satisfies the ``multiplicative lack of memory'' property, that is, $P(X+1>ab|X+1>a)=P(X+1>b)$, for every $a,b>1$ \citep{galambos}. Bear in mind that ageing notions are typically location and scale independent, in particular, $X$ and $\sigma X+\mu$ are equivalent in terms of the convex transform order. In this regard, we may refer to the LL(1) location-scale family, determined by the CDF ${L}$ up to location and scale transformations.

We now discuss the applicability of the IOR model. First, note that the IOR class may contain bathtub distributions, characterised by a {``decreasing then increasing'' behaviour of the} HR, and are commonly employed in survival analysis and reliability ({see, for example,} \cite{marshall2007}, {or} \cite{glaser}). For {instance}, under some conditions on the parameters, the bathtub distributions defined by \cite{toppleone}, \cite{hjorth1980}, \cite{schabe} or \cite{haupt1997} are IOR. {This} class contains also many heavy-tailed distributions, such as the log-logistic, the Pareto, the Burr XII, the Fr\'{e}chet, the Student's $t$ (with shape parameter(s) larger than or equal to 1), the lognormal (under some conditions) and the Cauchy distributions.
It is not possible to establish {inclusion relations} between the IOR class and other relevant classes which contain the IHR class, such as the IHRA and the DMRL families \citep{marshall2007}. However, like the IHR, these latter classes require the existence of all moments, therefore they cannot be used as shape constraints to estimate heavy-tailed distributions, which is one of the main advantages of the IOR model. This discussion motivates us to study nonparametric inference on lifetime distributions using the IOR constraint, instead of the IHR one.

\section{An IOR estimator}
Since we are interested in random lifetimes, hereafter we will focus on the case in which the distribution of interest $F$ has support included in the interval $[0,\infty)$. Given an ordered sample $X_{(1)},\ldots,X_{(n)}$ from $F$, the empirical CDF is defined as $\F_n(x)=\frac1n\sum_{i=1}^{n}{\mathbf{1}_{X_{(i)} \leq x}}$, where $\mathbf{1}_A$ is the indicator of the event $A$. A realisation of the random process $\F_n$ will be denoted with $F_n$. The empirical CDF represents the most natural way of estimating $F$ without any prior information, as it satisfies several important properties{, in particular}, the Glivenko-Cantelli Theorem establishes strong and uniform convergence of $\F_n$ to $F$. Assume that $F\in\mathcal{F}_{IOR}$: we are interested in determining an estimator, say $\widetilde{\F}_n$, such that $\widetilde{\F}_n\in\mathcal{F}_{IOR}$ and $\widetilde{\F}_n$ converges strongly and uniformly to $F$.
A first intuitive solution would be estimating $\Lambda_F$ using its empirical counterpart $\Lambda_{\F_n}=L^{-1}\circ \F_n=\frac{\F_n}{1-\F_n}$ (note that $\Lambda_{\F_n}(X_{(i)})=\frac i{n-i}$) and then estimate $F$ using $L\circ (\Lambda_{\F_n})_{cx}$. However, the function $L^{-1}(p)=\frac{p}{1-p}$ is unbounded in {$[0,1)$}, accordingly, for every sample size, $\sup{_{x}} |\Lambda_{\F_n}{(x)}-\Lambda{(x)}|=\infty$, and the corresponding IOR estimator defined earlier would not have the desired property. Therefore, we will consider a somewhat more involved approach.

For some given CDF $F\in\mathcal{F}_{IOR}$, let us define the following integral transform
\begin{equation*}T_F^{-1}(p)=Z_F\circ F^{-1}(p)=\int_0^{F^{-1}(p)}(1-F(x))^2\,dx,\quad p\in[0,1],
\end{equation*}
where $Z_F(t)=\int_0^t(1-F(x))^2\,dx$, and the corresponding scale-free version $\overline{T}_F^{-1}(p)=\frac{T_F^{-1}(p)}{T_F^{-1}(1)}$. We shall focus on distributions such that $T_F^{-1}(1)=\int_0^{\infty}(1-F(x))^2\,dx<\infty$. In particular, $T^{-1}(1)<\infty$ if $1-F$ is integrable, that is, if $F$ has finite mean, because {$\int_0^{\infty}(1-F(x))^2\,dx\leq\int_0^{\infty}(1-F(x))\,dx={\E(X)}$. However, the existence of the mean is not necessary{:} for example, for the CDF $L$ we have $T_L^{-1}(1)=1$, although $L$ has infinite mean. The function $T^{-1}$ belongs to the family of \textit{generalized total time on test} transforms, which has been studied by \cite{barlowzwet}. Note also that the inverse of $T_F^{-1}$, namely $T_F=F\circ Z^{-1}_F$, is a CDF with support $[0,T_F^{-1}(1)]$.

It can be seen that\begin{equation}\label{oddsttt}(T^{-1})'(p)=Z'\circ F^{-1}(p)(F^{-1})'(p)=\frac{(1-F\circ F^{-1}(p))^2}{f\circ F^{-1}(p)},\end{equation}
so that $T^{-1}$ is concave, or, {equivalently,} $T$ is convex, if and only if $\frac{f(x)}{(1-F(x))^2}$ is increasing, that is, $F$ is IOR. Note that $T^{-1}$ is the identity function if and only if $F$ is an LL(1) distribution.

For $k=1,\ldots,n$, the empirical counterpart of $T_F^{-1}(\frac kn)$ is
\begin{equation*}
T^{-1}_{\F_n}\left(\frac{k}n\right)=\int_0^{X_{(k)}}(1-\F_n(t))^2\,dt=\frac1{n^2}\left(\sum_{i=1}^{k-1}(2n-2i+1)X_{(i)}+(n-k+1)^2 X_{(k)}\right),
\end{equation*}
where, in particular, $X_{(0)}:=0$, $T_{\F_n}^{-1}{(0)}:=0$, and
$T^{-1}_{\F_n}(1)=\frac1{n^2}\sum_{i=1}^{n}(2n-2i+1)X_{(i)}.$
Note that $T^{-1}_{\F_n}(\frac kn)$ can also be expressed as
$$
T^{-1}_{\F_n}\left(\frac kn\right)=\sum_{j=1}^k\left(\frac{n-j+1}n\right)^2(X_{(j)}-X_{(j-1)}).  $$
For {general} $p\in[0,1]$, $T_{\F_n}^{-1}(p)$ is defined, by linear interpolation, as the piecewise linear function joining the points $\left(\frac {k-1}n,\,T^{-1}_{\F_n}\left(\frac{k-1}n\right)\right)$ and $\left(\frac {k}n,\,T^{-1}_{\F_n}\left(\frac{k}n\right)\right)$, $k=1,\ldots,n$. Similarly, the estimator of $\overline{T}^{-1}$ is $\overline{T}^{-1}_{\F_n}=\frac{T^{-1}_{\F_n}(p)}{T^{-1}_{\F_n}(1)}$. These estimators are strongly uniformly consistent.
\begin{lemma}\
\begin{itemize}
\item $T^{-1}_{\F_n}$ ($\overline{T}^{-1}_{\F_n}$) converges strongly and uniformly to $T_F^{-1}$ ($\overline{T}_F^{-1}$) in $[0,1]$;
\item $T_{\F_n}$ ($\overline{T}_{\F_n}$) converges strongly and uniformly to $T_F$ ($\overline{T}_{F}$) in $[0,T_F^{-1}(1)]$ ($[0,1]$).
\end{itemize}
\end{lemma}
\begin{proof}
The strong uniform convergence of $T^{-1}_{\F_n}$ is a special case of Theorem 2.1 of \cite{barlowzwet}. Since $T^{-1}_{\F_n}$ is bounded and has bounded domain, its inverse function $T_{\F_n}$ is also strongly uniformly consistent. Since $T^{-1}_{\F_n}(1)$ converges a.s. to $T^{-1}_{F}(1)$, it is clear that the same results hold for $\overline{T}^{-1}_{\F_n}$ and, likewise, for $\overline{T}_{\F_n}$.
\end{proof}
Henceforth, we will abbreviate $T^{-1}_{\F_n}$ and $\overline{T}^{-1}_{\F_n}$ with $T^{-1}_n$ and $\overline{T}^{-1}_{n}$, respectively, whenever it will be clear from the context that they correspond to a sample from $\F_n$.

{Note that under} the IOR constraint, $T^{-1}$ {is} concave. Therefore, an IOR estimator may be constructed using $(T_n^{-1})_{cv}$, namely, the LCM of $T_n^{-1}$. {In fact, $({T}_n^{-1})_{cv}$ is a concave piecewise linear function, so it always has a right} derivative, {denoted as} $\partial_+({T}_n^{-1})_{cv}$, which, by construction, is a decreasing right-continuous step function. Accordingly, {the relation in} (\ref{oddsttt}) may be used to obtain an increasing estimator of the OR, defined as follows:
\begin{equation*}\widetilde{\lambda}_n(x)=\frac1{\partial_+({T}_n^{-1})_{cv}\circ \F_n(x)},\;\text{ for } x\in(0,X_{(n)}),\end{equation*}
whereas $\widetilde{\lambda}_n(x)=0$, for $x\leq0$, and $\widetilde{\lambda}_n(x)=+\infty$, for $x\geq X_{(n)}$. In other words, $\widetilde{\lambda}_n(x)$ is the reciprocal of the slope (from the right) of $({T}_n^{-1})_{cv}$, evaluated at the point $\F_n(x)$. Accordingly, the odds function is estimated by
\begin{equation}\label{ORestimator}\widetilde{\Lambda}_n(x)=\int_0^x{\widetilde{\lambda}_n(t)\,dt}=\int_0^x\frac 1{\partial_+(T_n^{-1})_{cv}\circ \F_n(t)}\,dt.
\end{equation}
Finally, an IOR estimator of $F$ is given by $$
\widetilde{\F}_n=L\circ \widetilde{\Lambda}_n=1-\frac1{1+\widetilde{\Lambda}_n}.
$$
$\widetilde{\F}_n$ is clearly IOR, and it is absolutely continuous, except in the right endpoint of its support, $X_{(n)}$, at which it has a jump. Correspondingly, the PDF $f$ can be estimated, in $(0,X_{(n)})$, as
\begin{equation}
\label{eq:PDF_est}
\widetilde{f}_n=\frac{\widetilde{\lambda}_n}{(1+\widetilde{\Lambda}_n)^2}.
\end{equation}
Using the methods described in the book of \cite{robertson1988}, it is possible to {establish the asymptotic properties of these estimators}.
\begin{theorem}\label{estimators}
{Assume that $F$ is IOR, then:
\begin{enumerate}
\item $\widetilde{\lambda}_n \longrightarrow \lambda$ almost surely and uniformly in $[0,x_0]$, for every $x_0$ such that $\lambda(x_0)<\infty$;
\item $\widetilde{\F}_n{\longrightarrow} F$ almost surely and uniformly in $[0,\infty)$;
\item $\widetilde{f}_n{\longrightarrow} f$ almost surely and uniformly in $[0,x_0]$, for every $x_0$ such that $\lambda(x_0)<\infty$.
\end{enumerate}}
\end{theorem}
\begin{proof}
{Following the approach as in Theorem 7.4.1 of \cite{robertson1988} we prove that
$\lambda(x_0^-)\leq\liminf_{n\rightarrow\infty} \widetilde{\lambda}_n(x_0)\leq \limsup_{n\rightarrow\infty} \widetilde{\lambda}_n(x_0)\leq \lambda(x_0^+)$.}
{Indeed, since} the LCM $(T_n^{-1})_{cv}$ is concave by construction, then, for some small $\epsilon>0${, we have that}
\begin{multline*}\frac{(T_n^{-1})_{cv}(\F_n(x_0)+\epsilon)-(T_n^{-1})_{cv}(\F_n(x_0))}{\epsilon}\leq {\partial_+(T_n^{-1})_{cv}\circ \F_n(x_0)}=  \frac1{\widetilde{\lambda}_n(x_0^+)}\\
\leq \frac1{\widetilde{\lambda}_n(x_0^-)}={\partial_-(T_n^{-1})_{cv}\circ \F_n(x_0)}\leq \frac{(T_n^{-1})_{cv}(\F_n(x_0))-(T_n^{-1})_{cv}(\F_n(x_0)-\epsilon)}{\epsilon},
\end{multline*}
for every $x_0>0$ and $\epsilon\in(0,\min(F_n(x_0),1-F_n(x_0)))$. $T^{-1}_n$ converges uniformly to $T^{-1}$ on $[0,1]$, {consequently, the} Marshall's inequality \citep[Exercise 3.1]{groeneboom2014} gives$$\sup_{p\in[0,1]}|(T^{-1}_n)_{cv}(p)-T^{-1}(p)|\leq \sup_{p\in[0,1]}|{T}^{-1}_n(p)-T^{-1}(p)|\longrightarrow 0,$$
with probability 1. Then, letting $n\longrightarrow \infty$ in the above inequality, we obtain
\begin{multline*}\frac{{T}^{-1}(F(x_0)+\epsilon)-{T}^{-1}(F(x_0))}{\epsilon}\leq \limsup_{n\rightarrow\infty} \frac{1}{\widetilde{\lambda}_n(x_0)}\\ \leq\liminf_{n\rightarrow\infty} \frac{1}{\widetilde{\lambda}_n(x_0)} \leq \frac{{T}^{-1}(F(x_0))-{T}^{-1}(F(x_0)-\epsilon)}{\epsilon},
\end{multline*}
which gives the claimed bound {allowing} $\epsilon\longrightarrow 0$.
In particular, $\widetilde{\lambda}_n$ is uniformly strongly consistent at the continuity points of $\lambda$. Now, if $\lambda$ is increasing and continuous on some compact set $[0,x]$, $\widetilde{\lambda}_n$ converges uniformly to $\lambda$ on $[0,x]$. {Setting $x=x_0$, we obtain part 1).}
Subsequently, the dominated convergence theorem (note that $\widetilde{\lambda}_n$ is bounded on $[0,x]$) implies $\lim_{n\rightarrow\infty}\int_0^y\widetilde{\lambda}_n(t)\,dt=\lim_{n\rightarrow\infty}\widetilde{\Lambda}_n(y)=\int_0^y{\lambda}(t)\,dt=\Lambda(y)$, for every $y\in[0,x]$. Therefore, since the CDF $L$ is uniformly continuous in $[0,\infty)$, {it follows that} $$\widetilde{\F}_n(x)=L\circ \widetilde{\Lambda}_n(x)\longrightarrow  L\circ{\Lambda}(x)= F(x)$$
almost surely, for every $x$. Finally, following the same arguments in the proof of the Glivenko-Cantelli theorem, we can ensure that $\widetilde{\F}_n\longrightarrow  F$ almost surely and uniformly in $[0,\infty)$. {Since $f(x)=\lambda(x)(1-F(x))^2$ and, writing $\widetilde{\Lambda}_n=\Lambda_{\widetilde{\F}_n}$, we may express $\widetilde{f}_n(x)=\widetilde{\lambda}_n(x)(1-\widetilde{\F}_n(x))^2$. Then, part 1) and part 2) imply that $\widetilde{f}_n\longrightarrow f$ strongly and uniformly in $[0,x_0]$.}
\end{proof}

\subsection{Simulations}
We illustrate the numerical performance of $\widetilde{\F}_n$ and $\F_n$ in terms of MSE, as is well known, $\text{MSE}(\F_n(x))=\frac1nF(x)(1-F(x))$. {With regard to $\widetilde{\F}_n$}, the MSE was simulated using the following distributions: {\textit{(i)}} log-logistic, $F(x)=\frac1{1+x^{-a}}$, $a>0$, hereafter LL($a$), which is IOR for $a>1$ and has a constant OR for $a=1$; {\textit{(ii)}} Weibull $F(x)=1-e^{-x^a}$, $a>0$, hereafter W($a$), which is IOR for $a\geq1$ (and also log-concave, IHR, IHRA and DMRL); {\textit{(iii)}} beta type II, $F(x)=\beta(\frac{x}{1+x};a,b)$, $a,b>0$, hereafter B2($a,b$), which is IOR for $a>1,b\geq 1$ ($\beta$ denotes the regularized incomplete beta function){; {\textit{(iv)}} the Haupt and Schabe's distribution \citep{haupt1997} $F(x)=\sqrt{a^2+(2 a+1) x}-a$, {$a>-1/2$}, $x\in[0,1]$, hereafter HS($a$), which is bathtub for $a\in(\frac12,1)$ and IOR at the same time. Note that scale parameters are conveniently not considered in our analysis, since ageing properties are scale invariant.} {Figure \ref{f1-1} illustrates the estimators $\tilde{\F}_n$ and $\F_n$ in the IOR case, using a small sample size. Differently, the behaviour of such estimators in the non-IOR case is discussed in the next section and depicted in Figure \ref{f3}.}

\begin{figure}
\centering
\includegraphics[scale=0.8]{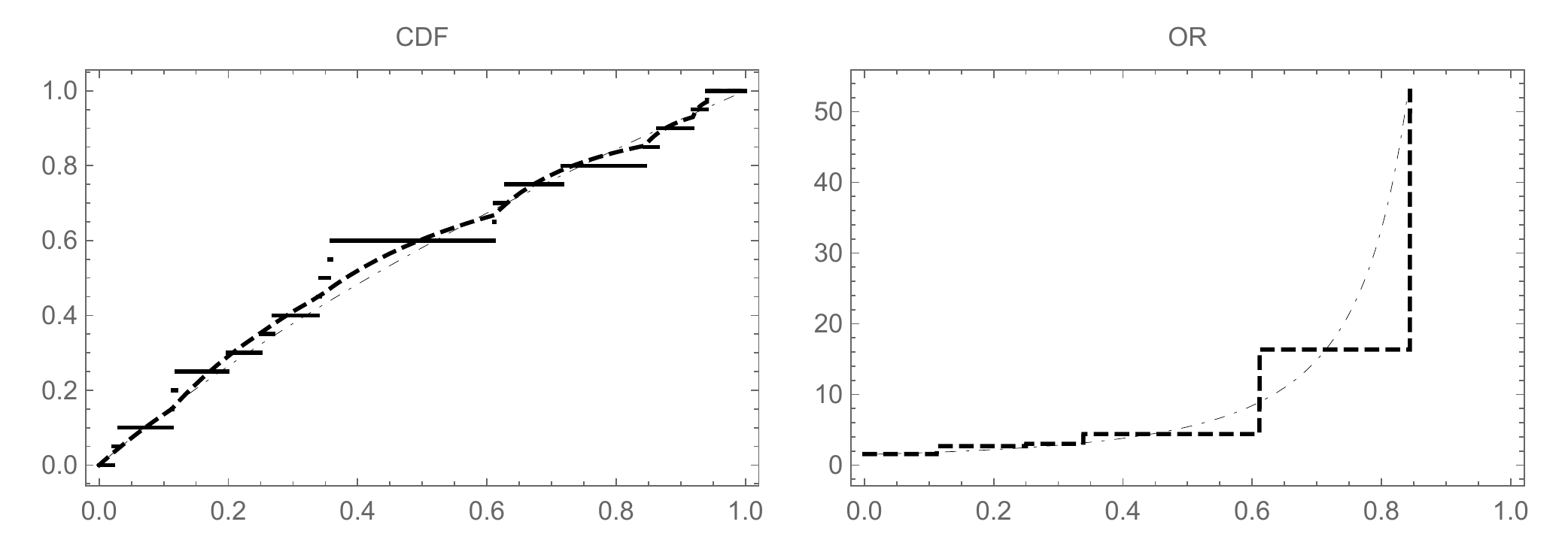}
\caption{Simulated small sample ($n=20$) from the IOR distribution HS(0.6). The plot on the left shows the CDFs $F$ (dot-dashed), $\F_n$ (solid) and $\widetilde{\F}_n$ (dashed); the plot on the right shows $\lambda_F$ (dot-dashed) and $\widetilde{\lambda}_n$ (dashed).\label{f1-1}}
\end{figure}

In Figure~\ref{f1} we plot simulated MSE values at all percentiles, based on 1000 runs, standardised by the corresponding ones for $\F_n${, that is, $\frac{\text{MSE}(\widetilde{\F}_n(x))}{\text{MSE}(\F_n(x))}$}. The simulation results show that, when $F$ is IOR, $\widetilde{\F}_n$ often outperforms $\F_n${, especially for extremely small deciles}. This is {particularly visible for} small sample sizes, as clearly when $n$ grows both estimators converge to $F$.
According to our results, the largest improvements often occur in the left tails, around deciles 0.1--0.3, however, in some cases $\widetilde{\F}_n$ performs worse at very small percentiles (0.01--0.05) for very small sample sizes ($n=10$) and some particular models, as it can be seen in Figure~\ref{f1}. Note that, among the distributions considered, only the W and the HS distributions have all moments, whereas the LL(1) has an infinite mean, the LL(2) and the B2 with $b\leq 2$ have infinite variance{,} and the B2 with $b\leq 3$ has an infinite third moment. This represents one of the main advantages of the proposed method, as the majority of the shape conditions which are commonly used in order-restricted inference under an ``adverse ageing scenario'' (log-concavity, IHR, IHRA, DMRL) require the existence of all moments. {Similarly, the HS($a$) is bathtub for the values of $a$ considered, so it is not compatible with the log-concavity, IHR, and IHRA assumptions (only the W($a$), for $a\geq1$, satisfies all these latter properties).} It can also be noted that convergence is somewhat slower for distributions with infinite variance. The case of the {LL(1)} is especially critical, because, besides having an infinite mean, more importantly, it has a constant OR, so we expect no improvement, or less, compared to the cases in which the OR is strictly increasing. {Surprisingly, even in this case, $\widetilde{\F}_n$ performs slightly better than $\F_n$ for smaller sample sizes, whereas, in general, $\widetilde{\F}_n$ performs better at the tails and $\F_n$ performs better around the median.}

\begin{figure}
\centering
\includegraphics[scale=0.8]{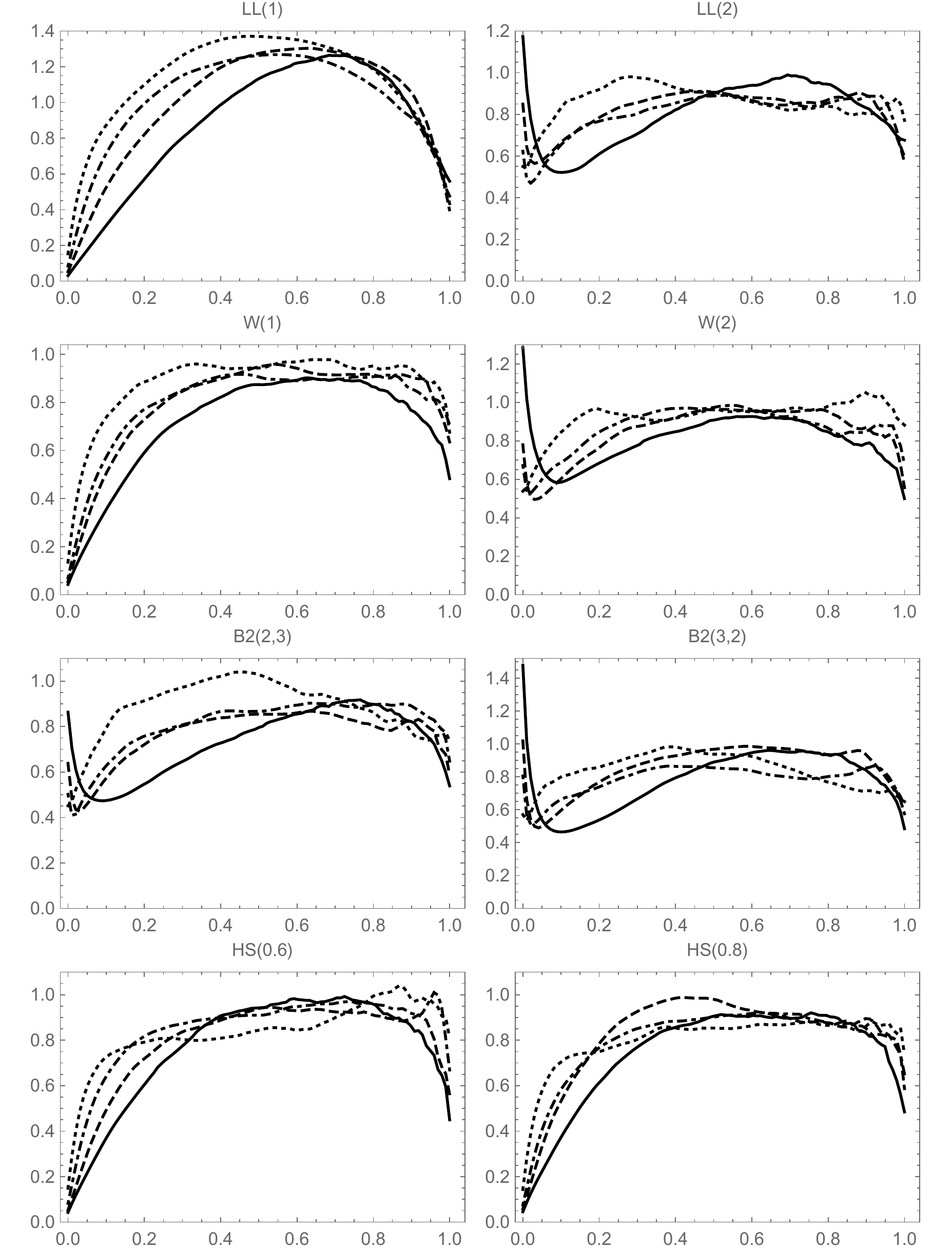}
\caption{{Standardised} MSE for $n=10$ (solid), $n=30$ (dashed), $n=50$ (dot-dashed), $n=100$ (dotted), evaluated at all percentiles.\label{f1}}
\end{figure}


\subsection{Smoothed estimators}
One may wonder whether the improvement of $\widetilde{\F}_n$ over ${\F}_n$ depends on the fact that $\widetilde{\F}_n$ is a smooth function. In response to this, {a numerical comparison of} $\widetilde{\F}_n$ with kernel estimators of the CDF, {may show} that the MSE of $\widetilde{\F}_n$ is often smaller even in this case, suggesting that the improvement of $\widetilde{\F}_n$ over ${\F}_n$ depends mostly on shape aspects, rather than on smoothness. Although the objective of this paper is to estimate the CDF, we will now briefly discuss the possibility of obtaining smooth estimators even for the PDF.
In fact, while $\widetilde{\F}_n$ is absolutely continuous in $[0,X_{(n)})$, the corresponding estimator of the PDF $\widetilde{f}_n$, {defined by (\ref{eq:PDF_est}),} is typically discontinuous and may exhibit spikes at the observed points (see Figure \ref{smooth}){. Hence,} it is generally not a good pointwise estimator of $f$. For this reason, it may be interesting to obtain a smooth estimator of the PDF which preserves the OR properties of $\widetilde{f}_n$, which may be achieved by smoothing the step function $\widetilde{\lambda}_n$ and then applying the approach discussed above.

Let {$k$} be {a zero-mean} PDF with support $[-1,1]$, for technical convenience. Given a bandwidth $h>0$ define as usual $k_h(x)=\frac1hk(\frac xh)$, and denote by $K_h$ the corresponding CDF{, that is, $K_h(x)=\int_{-\infty}^x k_h(t)\,dt$}. A smooth version of $\widetilde{\lambda}_n$ is given by
$$
\lambda_{n,h}^s(x)=\int k_{h}(x-t)\widetilde{\lambda}_n(t)\,dt=\int (1-K_{h}(x-t))\,d\widetilde{\lambda}_n(t){,}
$$
{after integration} by parts. Therefore, as $\widetilde{\lambda}_n$ is increasing by construction, it follows immediately that $\lambda^s_{n,h}$ is also increasing. Note that, since $\widetilde{\lambda}_n(x)=+\infty$, for $x\geq X_{(n)}$, then $\lambda_{n,h}^s(x)=+\infty$, for $x\geq X_{(n)}-h$, hence the smoothing has an effect just at points smaller than $X_{(n)}-h$. Now, a convex estimator of $\Lambda$ {is} given by $\Lambda^s_{n,h}=\int_{-\infty}^x\lambda_{n,h}^s(t)\,dt$, {where, unlike (\ref{ORestimator}), here the integration starts at $-\infty$, meaning that it assigns positive mass to the negative half line (this may be adjusted, if needed). Correspondingly,} an IOR estimator of the CDF {is defined as} $\F^s_{n,h}=L\circ \Lambda_{n,h}^s$, and a smooth estimator of the PDF {is given by} $f^s_{n,h}=\frac{\lambda_{n,h}^s}{(1+\Lambda_n^s)^2}$.

{Since $k$ has support $[-1,1]$, $\lambda_{n,h}^s(x)=\int_{x-h}^{x+h} k_{h}(x-t)\widetilde{\lambda}_n(t)\,dt$.
Therefore, by monotonicity, $\widetilde{\lambda}_n(t)\in[\widetilde{\lambda}_n(x-h),\widetilde{\lambda}_n(x+h)]${, which holds} for $t\in[x-h,x+h]$, implies that $\lambda_{n,h}^s(x)\in [\widetilde{\lambda}_n(x-h),\widetilde{\lambda}_n(x+h)]$. Then, if $h\longrightarrow 0$ as $n\longrightarrow\infty$, the asymptotic behaviour of $\widetilde{\lambda}_n${, proved in Theorem~\ref{estimators},} entails that $\lambda_{n,h}^s\longrightarrow \lambda$ strongly and uniformly in $[0,x_0]$ for every $x_0$ such that $\lambda(x_0)<\infty$. So, {using} the same arguments as in the proof of Theorem \ref{estimators}, $\F_n^s\longrightarrow F$ uniformly in $[0,\infty)$ and $f_n^s\longrightarrow f$ strongly and uniformly in $[0,x_0]$.

As usual in kernel estimation problems, the optimal choice of the bandwidth is the crucial issue. We are not dealing with this problem in the present paper, as the involved construction of $\widetilde{\lambda}_n$ makes it very hard to derive expressions for the MSE of $\lambda_{n,h}^s$. This estimator is also quite demanding from a computational point of view. As an illustration, Figure \ref{smooth} shows $\lambda_{n,h}^s$ and $f_{n,h}^s$ for a sample of size $n=20$ from a B2(2,3) distribution, using the Epanechnikov kernel with bandwidth $h=\frac 14$, that is, $k_{1/4}(x)=\frac{3}{4} \left(1-16 x^2\right)$, $x\in[-\frac 14,\frac 14]$.

\begin{figure}
\centering
\includegraphics[scale=0.8]{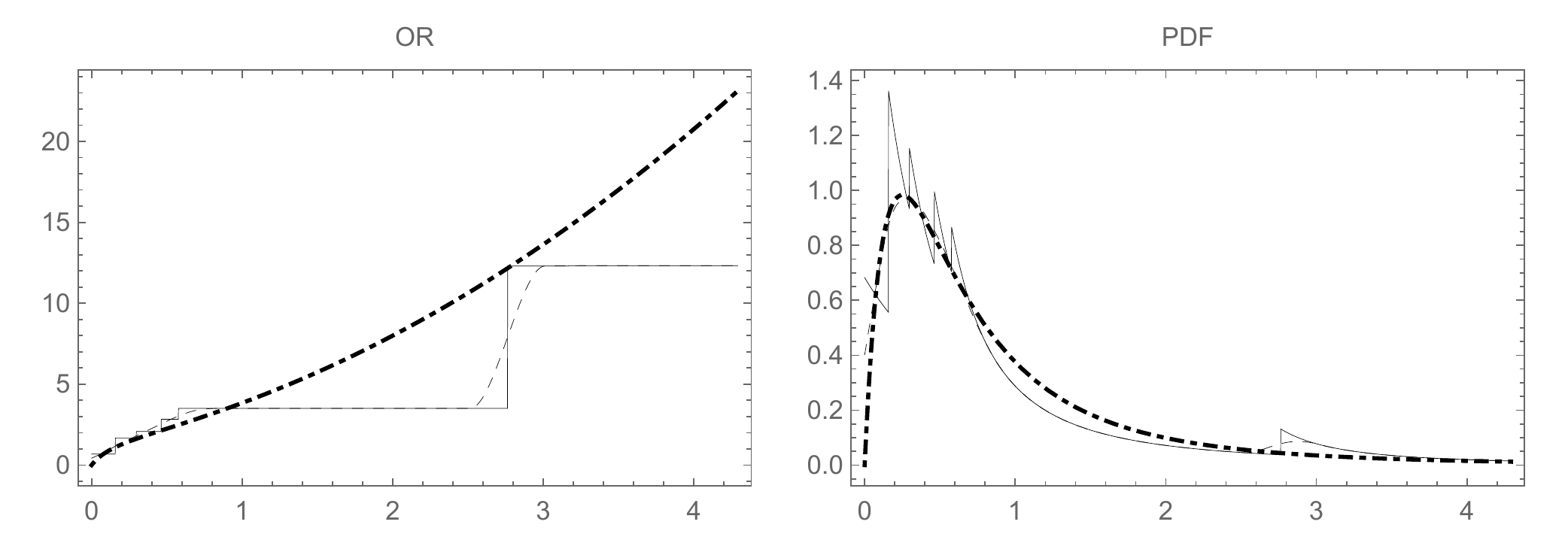}
\caption{Simulated small sample ($n=20$) from the IOR distribution B2(2,3). The plot on the left shows the ORs $\lambda_F$ (dot-dashed) $\widetilde{\lambda}_n$ (solid) and $\lambda_{n,h}^s$ (dashed), for $h=\frac14$; the plot on the right shows $f$ (dot-dashed), $\widetilde{f}_n$ (solid) and $f_{n,h}^s$ (dashed).\label{smooth}}
\end{figure}

\section{Tests for the IOR null hypothesis}
Tests of different ageing properties{, or, as discussed above, shape restrictions,} have been discussed extensively in the literature{: among many other authors,} \cite{proschan1967}, \cite{barlow1969}, \cite{bickel1969} or \cite{bickeldoksum} consider the null hypothesis of exponentiality versus the IHR alternative; \cite{deshpande}, \cite{kochar1985testing}, \cite{link}, \cite{wells}, and \cite{ahmad}, consider the null hypothesis of exponentiality versus the IHRA alternative; while \cite{tenga1984}, \cite{hall2005}, \cite{durot}, \cite{groeneboom2012}, \cite{gibels} or \cite{landoihr} consider the IHR null hypothesis versus the alternative that the HR is non-monotone.
This approach, {that is, setting the shape restriction as the null hypothesis}, may be more demanding from a computational point of view, but it has several advantages in controlling the behaviour of the test. We follow this idea and study tests of the null hypothesis $\mathcal{H}_0:F\in\mathcal{F}_{IOR}$ versus ${\mathcal{H}_1}:F\notin\mathcal{F}_{IOR}$. Tests of this kind have already been studied in \cite{oddspaper}, who focused on a restricted null hypothesis $\mathcal{H}_0^\nu:\text{``}\lambda\text{ is increasing in }S_\nu\text{''}$, where $S_\nu=\{x:x\leq F^{-1}(1-\nu)\}$, $\nu\in(0,1)$. {The domain restriction was {necessary} to control the {behaviour of the test statistic, given by the maximum }distance between the empirical odds function $\Lambda_{\F_n}=\frac{\F_n}{1-\F_n}$ and its GCM {$(\Lambda_{\F_n})_{cx}$}, due to the lack of uniform strong convergence of $\Lambda_{\F_n}$ {and $(\Lambda_{\F_n})_{cx}$} to $\Lambda_F$} {under $\mathcal{H}_0$}. {Differently, transporting the testing problem to the maximum distance between the estimators $\overline{T}_n$, or ${\F}_n$, {and their corresponding constrained versions $(\overline{T}_n)_{cx}$, or $\widetilde{\F}_n$, respectively, }means {that, under $\mathcal{H}_0$, }we do have uniform strong convergence in the whole domain, hence {restrictions are no} longer necessary}.

\subsection{A test based on the convexity of  $\overline{T}$}
Similarly to \cite{tenga1984}, who deal with the IHR property, we may obtain a first test which detects departures from the IOR property by considering a suitable distance between $\overline{T}_n$ and $(\overline{T}_n)_{cx}$, or, similarly, $\overline{T}^{-1}_n$ and $(\overline{T}_n^{-1})_{cv}$. We will consider the following {scale-independent} test statistic
\begin{equation*}
\KT(\F_n)=\sup_{u\in[0,1] } \vert \overline{T}_n(u)-(\overline{T}_n)_{cx}(u) \vert=\max_{1\leq i\leq n } \vert \frac in-(\overline{T}_n)_{cx}\circ \overline{T}_n^{-1}(\tfrac in) \vert.
\end{equation*}
{It} can be seen that, if the jump points of the empirical odds function $\Lambda_{\F_n}$, determined by the coordinates $(X_{(i)},\frac i{n-1})$, $i=1,\ldots,n-1$, lay on a convex curve, then $\KT(\F_n)=0$.

To obtain a conservative test, the determination of the least favourable distribution of $\KT$ under the null hypothesis is especially critical. As established in the following theorem, owing to a stochastic ordering result, such a distribution coincides with the distribution of $\KT$ in the case when the data are randomly sampled from the LL(1). To understand this behaviour, we need to introduce some additional notations.
Let us denote $\overline{Z}_F=\frac{Z_F}{T^{-1}_F(1)}$, so that $\overline{T}_F=F\circ \overline{Z}_F^{-1}$, and define the functional$$\Phi_p(F)=(F\circ \overline{Z}_F^{-1})_{cx}\circ \overline{T}_F^{-1}(p)=(\overline{T}_F)_{cx}\circ \overline{T}_F^{-1}(p).$$
The empirical counterpart of $\Phi_p(F)$, for $p=\frac in$, is $$\Phi_{i/n}(\F_n)=(\F_n\circ \overline{Z}_{\F_n}^{-1})_{cx}\circ \overline{T}_{\F_n}^{-1}(\tfrac in)=(\overline{T}_{\F_n})_{cx}\circ \overline{T}_{\F_n}^{-1}(\tfrac in).$$
Let us recall that $X$ is larger than $Y$ in the usual stochastic order, denoted as $X\geq_{st} Y$, if $P(X\geq t)\geq P(Y\geq t)$, for every $t$ \citep{shaked2007}. We can now establish the following result.
\begin{theorem}\label{order}
Let $\mathbb{L}_n$ be the empirical CDF corresponding to a random sample from $L$. Under $\mathcal{H}_0$, $\KT(\mathbb{L}_n)\geq_{st} \KT(\F_n)$.
\end{theorem}
\begin{proof}
Let $L_n$ be a realization of $\mathbb{L}_n$, corresponding to the observed sample $(y_{(1)},\ldots,y_{(n)})$. Since ${L}\geq_c F$, Theorem 1 of \cite{sjs} implies that $\Phi_{i/n}(\F_n)\geq_{st} \Phi_{i/n}(\mathbb{L}_n)$, for $i=1,\ldots,n$, provided that
\begin{equation}\label{functional}
(\overline{T}_{F^*_n})_{cx}\circ \overline{T}_{F^*_n}^{-1}(\tfrac in)=\Phi_{i/n}(F_n^*)\geq \Phi_{i/n}({L}_n)=(\overline{T}_{L_n})_{cx}\circ \overline{T}_{L_n}^{-1}(\tfrac in),
\end{equation}
where $F_n^*=F^{-1}\circ L\circ L_n=\tau\circ L_n$ is the empirical CDF corresponding to the values $\tau(y_{(i)})$, which determine an ordered sample from $F$.
Let
\begin{equation*}
u_i\!=\!\overline{T}^{-1}_{L_n}(\tfrac in)\!=\!\frac{\sum_{j=1}^i(\frac{n-j+1}n)^2(y_{(j)}-y_{(j-1)})}{\sum_{j=1}^n(\frac{n-j+1}n)^2(y_{(j)}-y_{(j-1)})}\;\mbox{and}\;u^*_i\!=\!\overline{T}_{F^*_n}^{-1}(\tfrac in)\!=\!\frac{\sum_{j=1}^i(\frac{n-j+1}n)^2(\tau(y_{(j)})-\tau(y_{(j-1)}))}{\sum_{j=1}^n(\frac{n-j+1}n)^2(\tau(y_{(j)})-\tau(y_{(j-1)}))}.
\end{equation*}
The denominators $\theta=\sum_{j=1}^n(\frac{n-j+1}n)^2(y_{(j)}-y_{(j-1)})$ and $\eta=\sum_{j=1}^n(\frac{n-j+1}n)^2(\tau(y_{(j)})-\tau(y_{(j-1)}))$ are constant because the samples are fixed. {We now define} the function $h:[0,1]\rightarrow[0,1]$ by
\begin{equation*}h(u_i)=h\left(\frac1{\theta        }{\sum_{j=1}^i(\frac{n-j+1}n)^2(y_{(j)}-y_{(j-1)})}\right)=\frac1{\eta}{\sum_{j=1}^i(\frac{n-j+1}n)^2(\tau(y_{(j)})-\tau(y_{(j-1)}))}=u^*_i,\end{equation*}
and by linear interpolation between the $u_i$'s values. The increments may be expressed as $u_i-u_{i-1}=\frac 1\theta (\frac{n-j+1}n)^2 (y_{(i)}-y_{(i-1)})$ and $u^*_i-u^*_{i-1}=\frac 1\eta(\frac{n-j+1}n)^2(\tau(y_{(i)})-\tau(y_{(i-1)}))$, respectively. To prove (\ref{functional}), note that, under $\mathcal{H}_0$, the function $\tau$ is increasing and concave, hence
\begin{equation*}\frac{h(u_{i+1})-h(u_{i})}{u_{i+1}-u_{i}}=\frac\theta\eta\frac{\tau(y_{i+1})-\tau(y_{i})}{y_{i+1}-y_{i}}
\leq\frac\theta\eta\frac{\tau(y_{i})-\tau(y_{i-1})}{y_{i}-y_{i-1}}=\frac{h(u_{i})-h(u_{i-1})}{u_{i}-u_{i-1}}.\end{equation*}
Therefore the slopes of the piecewise linear function $h$ are decreasing, that is, $h$ is concave. Now, Theorem 2.2 of \cite{tenga1984} yields (\ref{functional}). Subsequently, as $\Phi_{i/n}(\F_n)\geq_{st} \Phi_{i/n}(\mathbb{L}_n)$, $\KT(\F_n)=\max_i (\frac in-\Phi_{i/n}(\F_n))$ and similarly $\KT(\mathbb{L}_n)=\max_i (\frac in-\Phi_{i/n}(\mathbb{L}_n))$, we obtain the desired result.
\end{proof}

Theorem \ref{order} states that the random variable $\KT(\mathbb{L}_n)$ represents a stochastic (upper) bound for $\KT(\F_n)$ under the null hypothesis. Accordingly, the problem boils down to testing $\mathcal{H}_0^{L}:$ ``$F$ is LL(1)'' against $\mathcal{H}_1$, so we reject $\mathcal{H}_0$ when $\KT(F_n) \geq c_{\alpha,n}$, where $c_{\alpha,n}$ is the solution of $P(\KT(\mathbb{L}_n)\geq c_{\alpha,n}) = \alpha$. Theorem \ref{order} ensures that the probability of rejecting $\mathcal{H}_0$ when it is true, is at most $\alpha$, that is, the test has size $\alpha$ under $\mathcal{H}_0$. {Similarly, if $F$ is DOR, it is easy to see that $P(\KT(\F_n)\geq c_{\alpha,n})\geq P(\KT(\mathbb{L}_n)\geq c_{\alpha,n})=\alpha$, that is, the test is \textit{unbiased} under DOR alternatives.} For a given realisation $F_n$, the $p$-value of the test is $p=P(\KT(\mathbb{L}_n)\geq \KT(F_n))$.

As established by the following theorem, $\KT$ is also capable of detecting any deviation from the IOR null hypothesis.
\begin{theorem}\label{consistency}
Under $\mathcal{H}_1$, $\lim_{n\rightarrow \infty}P(\KT(\F_n)>c_{\alpha,n})=1$.
\end{theorem}
\begin{proof}
If $\mathcal{H}_0$ is true, considering the special case when $\overline{T}_n$ is obtained by sampling from the LL(1), $\overline{T}_n$ and $(\overline{T}_n)_{cx}$ converge strongly and uniformly to the identity function. In particular, Marshall's inequality gives $\sup_{u\in[0,1]}|(\overline{T}_n)_{cx}(u)-u|\leq \sup_{u\in[0,1]}|{\overline{T}}_n(u)-u|,$ with probability 1. Then, for every fixed $\alpha\in(0,1)$, there exists some $n_0$ such that, for $n>n_0$
\begin{equation*}
P(\sup_u|\overline{T}_n(u)- u|\leq \frac\epsilon2\land\sup_u|(\overline{T}_n)_{cx}(u)- u|\leq \frac\epsilon2)\geq 1-\alpha
\end{equation*}
However, the function $u-\frac\epsilon2$ is convex, therefore, for $n>n_0,$ $u-\frac\epsilon2\leq(\overline{T}_n)_{cx}(u)\leq \overline{T}_n(u)\leq u+\frac\epsilon2,\forall u$, and by inclusion we obtain $$P(-\epsilon\leq \overline{T}_n(u)- (\overline{T}_n)_{cx}(u)\leq \epsilon,\forall u)=P(\sup_u( \overline{T}_n(u)- (\overline{T}_n)_{cx}(u))\leq \epsilon)\geq 1-\alpha.$$
Since $\epsilon$ can be arbitrarily small, $c_{\alpha,n}\longrightarrow0$ for $n\longrightarrow\infty$.

Suppose that $\mathcal{H}_1$ is true. Then $d=\sup_u (\overline{T}(u)-(\overline{T}_n)_{cx}(u))>0$. Moreover, $\overline{T}_n$ converges strongly and uniformly to $\overline{T}$, whereas $(\overline{T}_n)_{cx}$ converges strongly and uniformly to $\overline{T}_{cx}$ (this can be seen using the same argument as in Theorem 3 of \cite{landoihr}). Therefore, given some $\epsilon>0$, there exists some $n_0$ such that, for $n>n_0$, $\sup_u|\overline{T}(u)-\overline{T}_n(u)|<\frac\epsilon2$ and $\sup |(\overline{T}_n)_{cx}(u)-\overline{T}_{cx}(u)|<\frac\epsilon2$, with probability 1. Then, for $n>n_0$\begin{equation*}
\overline{T}_n(u)-(\overline{T}_n)_{cx}(u)> \overline{T}(u)-\frac\epsilon2-(\overline{T}_{cx}(u)+\frac\epsilon2)=\overline{T}(u)-\overline{T}_{cx}(u)-\epsilon
\end{equation*}
almost surely, for every $u$, which implies
\begin{equation*}
\sup_u(\overline{T}_n(u)-(\overline{T}_n)_{cx}(u))> \sup_u( \overline{T}(u)-\overline{T}_{cx}(u)-\epsilon)=d-\epsilon>0.
\end{equation*}
Therefore, since $\epsilon$ can be arbitrarily small, $P(\sup_u(\overline{T}_n(u)-(\overline{T}_n)_{cx}(u))\geq d)\longrightarrow 1$, for $n\longrightarrow\infty$. But since $c_{\alpha,n}\rightarrow0$, then $P(\KT(\F_n)>c_{\alpha,n})\longrightarrow1$.
\end{proof}
It is important to remark that one may replace the sup-norm with a different type of distance between $(\overline{T}_n)_{cx}$ and $\overline{T}_n$ and still obtain a consistent test, since the key steps in the proof of Theorem \ref{consistency} are the uniform consistency of these two estimators under $\mathcal{H}_0$, while only $\overline{T}_n$ converges to $\overline{T}$ under $\mathcal{H}_1$.

\subsection{A test based on $\widetilde{\F}_n$}
A second test can be obtained by considering a suitable distance between $\F_n$ and $\widetilde{\F}_n$, such as the uniform norm. Let us define the following Kolmogorov-Smirnov type test statistic
\begin{equation*}
\KS({\F}_n)=\sup_{x>0} \vert \F_n(x)-\widetilde{\F}_n(x) \vert.
\end{equation*}

About the IHR property, a Kolmogorov-Smirnov test of this type has been studied by \cite{landoihr}. Note that $\KS$ is also scale-independent. However, if the jump points of $\Lambda_{\F_n}$ lay on a convex curve, then $\KS(\F_n)$ does not coincide with 0, differently from $\KT(\F_n)$. To fulfil this property, one should consider a modified version of $\KS$, say $\widehat{\KS}(\F_n)=\sup_{x>0} \vert \widehat{\F}_n(x)-\widetilde{\F}_n(x) \vert$, where $\widehat{\F}_n$ has the same construction of $\widetilde{\F}_n$, as described in Section 3, with the exception that it does not include any shape constraint, namely, $$\widehat{\F}_n(x)=L\left(\int_0^x\frac 1{\partial_+(T_n^{-1})\circ \F_n(t)}dt\right).$$
$\widehat{\F}_n$ and $\F_n$ are very similar and it is easy to see that they are asymptotically equivalent, moreover, the computation of $\widehat{\F}_n$ is quite demanding without having an improvement on the power of the test. Therefore, henceforth we will focus only on $\KS$.

As in the previous subsection, the distribution of $\KS$ may be determined by simulating from $L$. However, in this case, the stochastic ordering arguments used in Theorem \ref{order} do not hold, because of the complicated construction of $\widetilde{\F}_n$. Accordingly, we cannot formally establish a stochastic bound for the size of the test under $\mathcal{H}_0$, so it is more correct to present this as a test for $\mathcal{H}^L_0$ against $\mathcal{H}_1$. We reject the null $\mathcal{H}^L_0$ when $\KS(F_n) \geq c_{\alpha,n}$, where $c_{\alpha,n}$ is the solution of $P(\KS(\mathbb{L}_n)\geq c_{\alpha,n}) = \alpha$. However, even in some critical IOR cases, numerical evidence, reported in the next section, shows that the type-I error of $\KS$ is always smaller than $\alpha$. Moreover, the test is consistent against non-IOR alternatives, as established by the following theorem. The proof is omitted because it can be obtained using the same arguments as in the proof of Theorem 3 of \cite{landoihr}.

\begin{theorem}\label{consistency2}
Under $\mathcal{H}_1$, $\lim_{n\rightarrow \infty}P(\KS(\F_n)>c_{\alpha,n})=1$.
\end{theorem}
Similarly to what has been discussed in the previous subsection, one may replace the sup-norm with a different type of distance between $\widetilde{\F}_n$ and ${\F}_n$, and the corresponding test would still be consistent.

\subsection{Simulations}
We compare the performance of the two tests proposed simulating from some popular parametric families of distributions. As these tests are scale invariant, we set the scale parameters to 1. Some special cases of interest are {\textit{(i)}} IOR models{, \textit{(ii)}} decreasing OR (DOR) models{, \textit{(iii)}} non-monotone OR models. As for {\textit{(i)}}, we consider the LL($a$), with shape parameter $a$ ranging in the interval $[1,1.2]$. {With regard to \textit{(ii)}}, we consider the LL($a$) with {$a\in[0.7,1)$}. These cases are especially difficult to detect because, for $a=1$, the OR is constant. As for {\textit{(iii)}}, we consider the W($a$) with shape parameter {$a\in[0.3,0.8]$}, which exhibits a decreasing-increasing OR for $a<1$; the B2 distribution, which has an increasing-decreasing OR for {$a<1$ and $b>1$}, (thus we consider $b=2$, $a\in[0.3,0.7]$); and the Birnbaum-Saunders (BS) distribution with CDF $\Phi\left(\frac1{a}\left(\sqrt x-\frac{1}{\sqrt x}\right) \right)$, $x,a>0$ ($\Phi$ denotes the standard normal CDF), which has an increasing-decreasing-increasing OR for $a\in[2,4]$ (this is especially critical to detect for smaller values of $a$). The results are reported in Figure \ref{power}, which show the rejection rates at level $\alpha=0.1$, corresponding to 500 simulation runs. The plots show that the performance of the tests is very similar when $\mathcal{H}_0$ is true (LL($a$) with $a\geq1$) and against DOR alternatives (LL($a$) with $a\leq 1$). In particular, it can be seen that, for both tests, in the IOR case the simulated type-I error probability is always bounded by $\alpha$, whereas, in the DOR case, the simulated power is always greater than $\alpha$. About $\KT$, this is formally established by Theorem \ref{order}, which represents an advantage. However, it can be seen from Figure \ref{power} that $\KS$ remarkably outperforms $\KT$ when $F$ has a non-monotone OR, which is typically the most critical case to detect. {For both tests, the simulated power increases with the sample size, confirming the consistency properties established in Theorem \ref{consistency} and Theorem \ref{consistency2}.}

\begin{figure}
\includegraphics[scale=0.8]{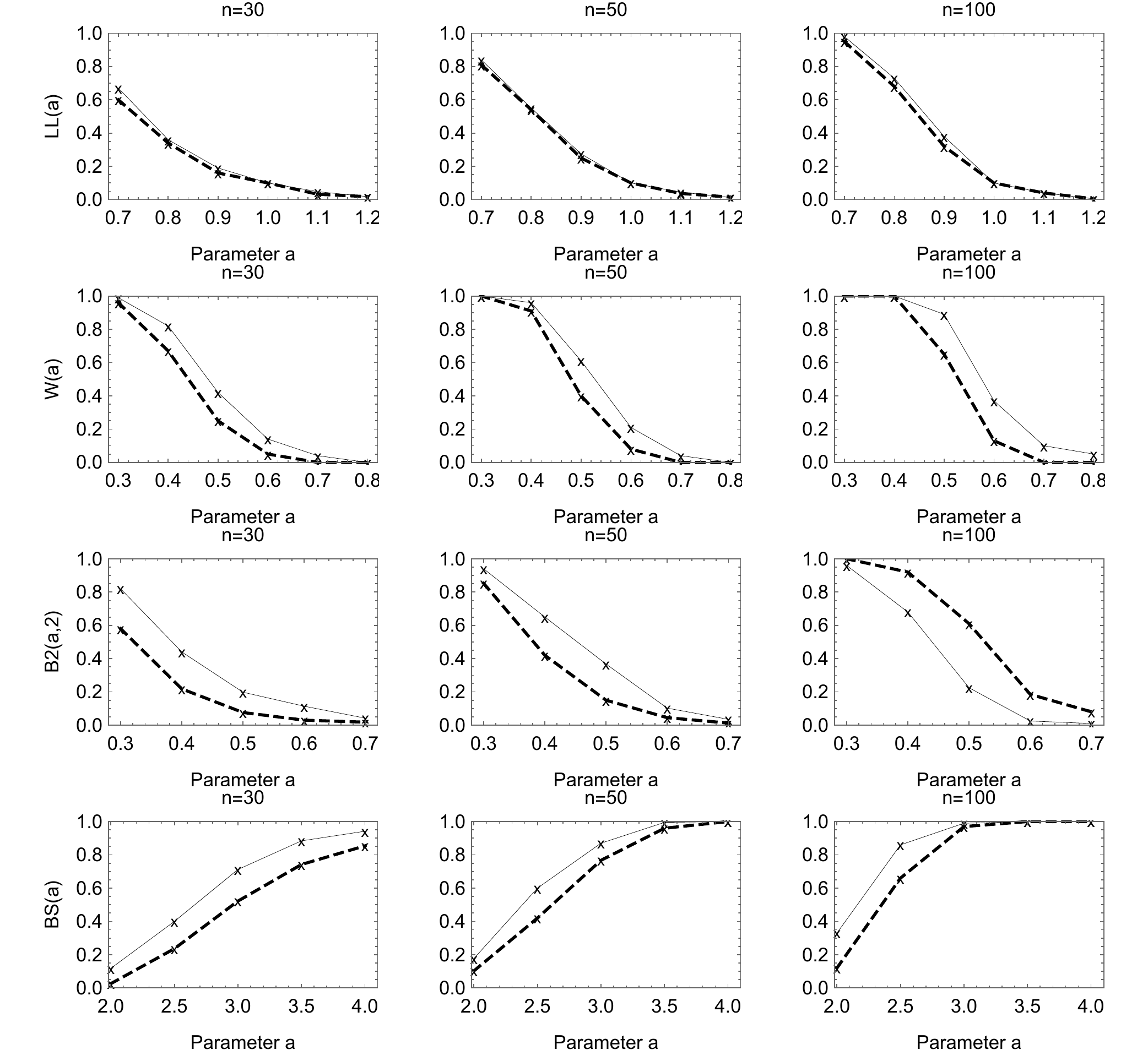}
\caption{Simulated power of $\KS$ (solid) and $\KT$ (dashed) under LL (first row), W (second row), B2 (third row), and BS (fourth row) alternatives.\label{power}}
\end{figure}

To illustrate the behaviour of $\KT$ and $\KS$ under non-IOR alternatives, we also consider the following CDF, $$F_{a,b}(x)=\begin{cases}
 \frac{x^a}{x^a+1} & 0<x\leq 1, \\
 \frac{x^b}{x^b+1} & x>1,
\end{cases}\qquad a,b>0.$$
For $a,b\geq 1$, the OR of this distribution is increasing almost everywhere, however, for $a>b$, it has a downward jump at $x=1$, so that $F_{a,b}$ cannot be IOR. Figure \ref{f3} shows the remarkable distance between $\overline{T}_n$ and $(\overline{T}_n)$, and between $\F_n$ and $\widetilde{\F}_n$, for a simulated sample of size 100 from $F_{5,1}$. In both cases, these large distances lead to the rejection of $\mathcal{H}_0$.

\begin{figure}
\includegraphics[scale=0.8]{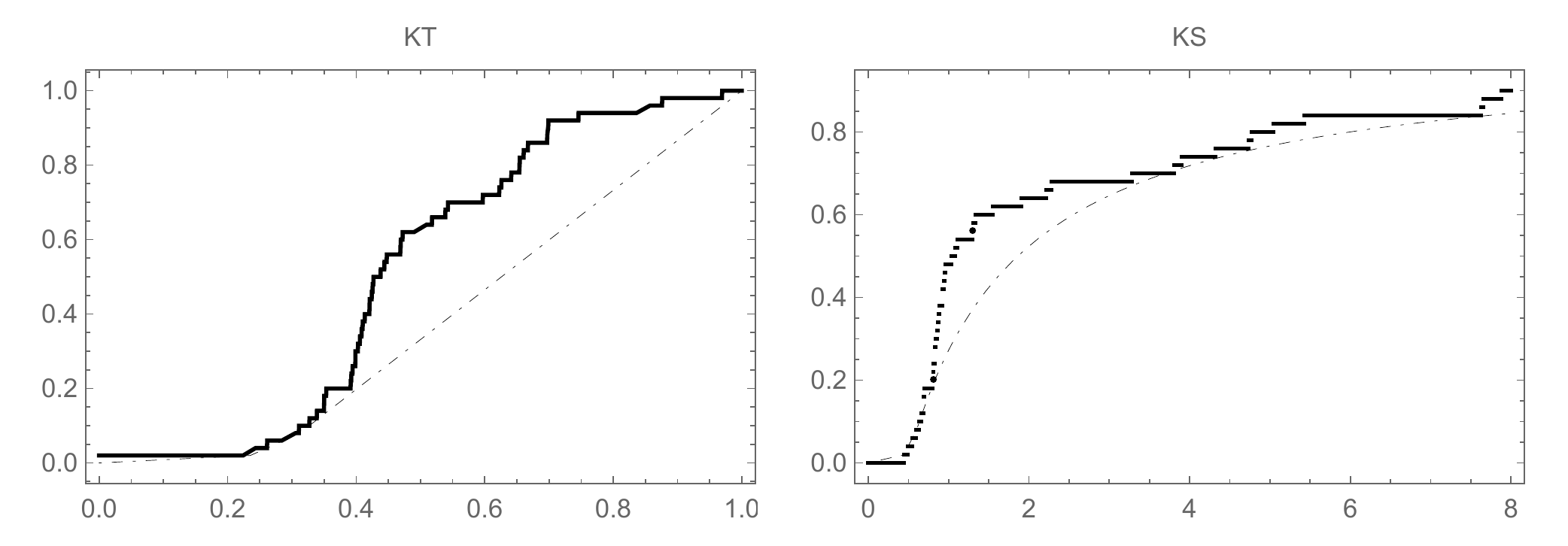}
\caption{Simulated sample ($n=50$) from the non-IOR CDF $F_{5,1}$. $\overline{T}_n$ and $(\overline{T}_n)$ on the left side; $\F_n$ and $\widetilde{\F}_n$ on the right side. In both cases, the two estimators (constrained and unconstrained) diverge for quantile values larger than 0.2, approximately.\label{f3}}
\end{figure}

\textbf{Funding.}{ T.L. was supported by the Italian funds ex MURST 60\% 2021, by the Czech Science Foundation (GACR) under project 20-16764S and V\v{S}B-TU Ostrava under the SGS project SP2021/15. I.A. and P.E.O. were partially supported by the Centre for Mathematics of the University of Coimbra UID/MAT/00324/2020, funded by the Portuguese Government through FCT/MCTES and co-funded by the European Regional Development Fund through the Partnership Agreement PT2020.}

\bibliographystyle{elsarticle-harv}
\bibliography{biblio}





\end{document}